%% file: Paper.tex
\documentclass[conference]{IEEEtran}
\IEEEoverridecommandlockouts
\usepackage{cite}
\usepackage{amsmath,amssymb,amsfonts}
\usepackage{algorithm}
\usepackage{algpseudocode}

\usepackage{amsthm}
\usepackage{graphicx}
\usepackage{textcomp}
\usepackage{xcolor}
\newtheorem{theorem}{Theorem}
\newtheorem{dfn}{Definition}
\newtheorem{expl}{Example}
\usepackage{subcaption}
\usepackage{multirow}

\newcommand{\msat}{MAX-3SAT }

\newcommand{\tsat}{3SAT}
\newcommand{\pq}{Pattern QUBO }

\newcommand{\sota}{state-of-the-art}

\def\BibTeX{{\rm B\kern-.05em{\sc i\kern-.025em b}\kern-.08em
    T\kern-.1667em\lower.7ex\hbox{E}\kern-.125emX}}
\begin{document}

\title{Solving Max-3SAT Using QUBO Approximation\\

}

\author{\IEEEauthorblockN{Sebastian Zielinski}
\IEEEauthorblockA{\textit{Institute for Informatics} \\
\textit{LMU Munich}\\
Munich, Germany \\
sebastian.zielinski@ifi.lmu.de}
\and
\IEEEauthorblockN{Jonas Nüßlein}
\IEEEauthorblockA{\textit{Institute for Informatics} \\
\textit{LMU Munich}\\
Munich, Germany \\
jonas.nuesslein@ifi.lmu.de}
\and
\IEEEauthorblockN{Michael Kölle}
\IEEEauthorblockA{\textit{Institute for Informatics} \\
\textit{LMU Munich}\\
Munich, Germany \\
michael.koelle@ifi.lmu.de}
\and
\IEEEauthorblockN{Thomas Gabor}
\IEEEauthorblockA{\textit{Institute for Informatics} \\
\textit{LMU Munich}\\
Munich, Germany \\
thomas.gabor@ifi.lmu.de}
\and
\IEEEauthorblockN{Claudia Linnhoff-Popien}
\IEEEauthorblockA{\textit{Institute for Infomatics} \\
\textit{LMU Munich}\\
Munich, Germany \\
linnhoff@ifi.lmu.de}
\and
\IEEEauthorblockN{Sebastian Feld}
\IEEEauthorblockA{\textit{Quantum \& Computer Engineering} \\
\textit{Delft University of Technology}\\
Delft, The Netherlands \\
s.feld@tudelft.nl}
}

\maketitle
\input{chapters/0-Abstract}

\begin{IEEEkeywords}
Quadratic Unconstrained Binary Optimization, Combinatorial Optimization, Max-3SAT, Approximation
\end{IEEEkeywords}

\input{chapters/1-Introduction}
\input{chapters/2-Foundations}

\input{chapters/3-RelatedWork}
\input{chapters/4-ApproximationMethods}
\input{chapters/5-Evaluation}
\input{chapters/6-Conclusion}
\section*{Acknowledgement}
This paper was partially funded by the German Federal Ministry of Education and Research through the funding program ``quantum technologies --- from basic research to market'' (contract number: 13N16196).

\bibliographystyle{IEEEtran}
\bibliography{references}
\end{document}

%% file: chapters/0-Abstract.tex
\begin{abstract}
As contemporary quantum computers do not possess error correction, any calculation performed by these devices can be considered an involuntary approximation. To solve a problem on a quantum annealer, it has to be expressed as an instance of Quadratic Unconstrained Binary Optimization (QUBO). In this work, we thus study whether systematically approximating QUBO representations of the \msat problem can improve the solution quality when solved on contemporary quantum hardware, compared to using exact, non-approximated QUBO representations. For a \msat instance consisting of a \tsat~ formula with  $n$ variables and $m$ clauses, we propose a method of systematically creating approximate QUBO representations of dimension ($n\times n)$, which is significantly smaller than the QUBO matrices of any exact, non-approximated \msat QUBO transformation.  In an empirical evaluation, we demonstrate that using our QUBO approximations for solving \msat problems on D-Wave's quantum annealer Advantage\_System6.4 can yield better results than using \sota~ exact QUBO transformations. Furthermore, we demonstrate that using naive QUBO approximation methods, based on removing values from exact $(n+m)\times(n+m)$-dimensional QUBO representations of \msat instances, is ineffective.
\end{abstract}

%% file: chapters/1-Introduction.tex
\section{Introduction}
Satisfiability problems are at the core of many theoretical and practical computer science applications. In theoretical computer science, they are often used to prove the hardness of other problem classes. In practical domains, they are used to solve problems in the verification of integrated circuits \cite{marques2000boolean}, planning \cite{rintanen2006planning,duan2012provable}, dependency resolution \cite{abate2020dependency}, and many more. Given a Boolean formula, the satisfiability problem is concerned with deciding whether a satisfying assignment of Boolean values to the formula's variables exists. The optimization version of the satisfiability problem is called the MAX-SAT problem. The MAX-SAT problem also consists of a Boolean formula, but the goal is to find an assignment that satisfies as many clauses as possible. As there is a polynomial-time transformation from any satisfiability problem to a \tsat{} problem, which is a satisfiability problem in which each clause consists of at most 3 literals, often \tsat{} problems are studied as a canonical representative of satisfiability problems. \\
Quantum computing is a computational paradigm that promises to speed up the search for solutions to certain problems (e.g., Shor's algorithm \cite{shor1994algorithms}, Grover's algorithm \cite{grover1996fast}). With the recent improvements in the availability and capabilities of quantum hardware systems, researchers began to explore the possibilities of solving \msat problems on quantum computers. One area of research studies methods of transforming \msat problems to instances of Quadratic Unconstrained Binary Optimization (QUBO), as QUBO is the input format for contemporary quantum annealers and for QAOA algorithms on quantum gate systems. In the last decade, numerous methods have been proposed to transform instances of \msat problems into instances of QUBO \cite{nusslein2023solving,chancellor2016direct,choi2010adiabatic,bian2020solving,zielinski23Pattern,verma2021efficient}. It has also been shown that choosing different transformations can significantly impact the solution quality \cite{zielinski2023influence, kruger2020quantum}. Consequently, finding new QUBO transformations that potentially increase a solver's capability of finding high-quality solutions for \msat problems is an active field of research. \\
When using contemporary quantum computers as QUBO solvers, one must consider a major challenge these devices currently face: missing error correction. As these devices currently do not possess error correction, calculations on these devices suffer from the influence of noise. Thus, it can be argued that any calculation of a contemporary quantum computer is an involuntary approximation of a given input problem. Hence, as the subtleties of a QUBO instance are lost during the calculation process due to the influence of noise, evaluating the possibilities of purposefully creating QUBO approximations of given problems as an input to quantum computers seems interesting. \\
To express a given problem as an instance of QUBO often requires many auxiliary variables to precisely express all the constraints and objectives of a real-world problem as a quadratic problem. Hence, the idea of a QUBO approximation is to create a simpler quadratic representation of a real-world problem that, when minimized, yields good (enough) solutions to the real-world problem. The goal of these simplifications is to decrease the QUBO size or its density (i.e., the dependencies between different variables). A smaller (resp. less dense) QUBO generally needs fewer physical qubits on quantum hardware. Furthermore, it seems reasonable to assume that reducing the dependencies between different variables might reduce the severity of local calculation errors introduced by noise. Thus, the question we want to address in this paper is: Considering the influence of noise on calculations of contemporary quantum hardware, can solving purposefully crafted QUBO approximations of \msat problems yield comparable or even better results than solving exact, non-approximated QUBO representations of \msat problems on contemporary quantum hardware? In other words, can we craft QUBO approximations for \msat instances such that the error we introduce by approximating the problem does not exceed the error the noise introduces into quantum calculations when solving larger, exact QUBO representations of the same instances? \\
While there are several studies concerning the question of creating QUBO approximations for given problems (e.g. \cite{nusslein2023black,matsumori2022application}), the motivation for these studies is that either the objective function is unknown and thus the QUBO representation is approximated using a black-box approach, or that the constraints cannot be expressed precisely as a quadratic optimization problem and thus need approximation. Our work, however, is motivated by the study of Sax et al. \cite{sax2020approximate}. In their study, the authors transformed satisfiability problems into instances of QUBO using Choi's method \cite{choi2010adiabatic}. Then, they employed a simple QUBO approximation method, namely removing entries from a given QUBO matrix, and observed that approximately 70\% of the entries could be removed without a significant decrease in the solution quality.\\
This work proposes a new method to create QUBO approximations of \msat problems systematically. Furthermore, we show that under the current circumstances (missing error correction of quantum hardware), we can receive competitive or even better solutions when solving QUBO approximations of \msat instances on D-Wave's quantum annealer Advantage\_System6.4 compared to solving exact, non-approximated QUBO representations of the same instances.

Our contributions in this paper are:
\begin{enumerate}

    \item[1.]  We introduce a new method of systematically creating effective QUBO approximations for \msat problems.
      \item[2.] In an empirical case study performed on D-Wave's quantum annealer Advantage\_System6.4, we demonstrate that our method of systematically creating QUBO approximations for \msat problems can yield comparable or even better results than non-approximated, exact QUBO transformations.
      \item[3.] We provide evidence for the scaling of this approach. That is, we show that our QUBO approximation method still yields good results when larger formulas, are solved via a classical QUBO solver.
    \item[4.]We  show that naive methods of creating QUBO approximations for the \msat problem, based on the removal of values from given $(n+m)\times(n+m)$-dimensional exact QUBO matrices, generally lead to declines in solution quality.

\end{enumerate}

The remainder of the paper is organized as follows: Section \ref{sec:foundations} introduces the \msat problem and transformations from \msat instances to instances of QUBO. In Section \ref{sec:related_work}, we state related work. In Section \ref{sec:app_methods} we present our method of systematically approximating QUBO transformations for \msat problems. In Sec. \ref{sec:eval} we perform empirical evaluations on quantum and classical hardware to show the potential of our systematic QUBO approximation method. We conclude the paper in Sec. \ref{sec:conclusion} and state future research opportunities.

%% file: chapters/2-Foundations.tex
\section{Foundations}\label{sec:foundations}
\subsection{Satisfiability Problems}
Satisfiability problems are concerned with the satisfiability of Boolean formulas. Thus, we will first define a Boolean formula:\\
\begin{dfn}[Boolean formula \cite{arora2009computational}]
Let $x_1,..., x_n$ be Boolean variables. A \emph{Boolean formula} consists of the variables $x_1, ..., x_n$ and the logical operators  $\wedge$,  $\vee$,  $\neg$. Let $z \in \{0,1\}^n$ be a vector of Boolean values. We identify the value 1 as TRUE and the value 0 as FALSE. The vector $z$ is also called an \emph{assignment}, as it assigns truth values to the Boolean variables $x_1, ..., x_n$ as follows: $x_i = z_i$, where $z_i$ is the $i-th$ component of $z$. If $\phi$ is a Boolean formula, and $z \in \{0,1\}^n$ is an assignment, then $\phi(z)$ is the evaluation of $\phi$ when the variable $x_i$ is assigned the Boolean value $z_i$. If there exists a $z \in \{0,1\}^n$, such that $\phi(z)$ is TRUE, we call $\phi$ satisfiable. Otherwise, we call $\phi$ unsatisfiable \cite{arora2009computational}.
\end{dfn}
Satisfiability problems are often given in conjunctive normal form, which we will define next:
\begin{dfn}[Conjunctive Normal Form \cite{arora2009computational}] A Boolean formula over variables $x_1,...,x_n$ is in \emph{Conjunctive Normal Form (CNF)} if it is of the following structure:
$$\bigwedge_i \big(\bigvee_{j} y_{i_j}\big)$$
Each $y_{i_j}$ is either a variable $x_k$ or its negation $\neg x_k$. The $y_{i_j}$ are called the \emph{literals} of the formula. The terms $(\vee_j y_{i_j})$ are called the \emph{clauses} of the formula. A \emph{$k$CNF} is a CNF formula, in which all clauses contain at most $k$ literals. 
\end{dfn}
Given a Boolean formula $\phi$ in \emph{$k$CNF}, the satisfiability problem is the task of determining whether $\phi$ is satisfiable or not. This problem was one of the first problems for which NP-completeness has been shown \cite{cook1971complexity}. In this paper, we will especially consider \emph{$3$CNF} problems, which we will refer to as \tsat{} problems.\\ The generalization of the SAT problem is called MAX-SAT. In the MAX-SAT problem, we are given a Boolean formula $\phi$ consisting of $m$ clauses. The task is to find an assignment of truth values to the variables of $\phi$ such that as many clauses as possible are satisfied. Finding an assignment in the MAX-SAT problem that satisfies $m$ clauses is thus equivalent to solving the corresponding satisfiability problem (i.e., determining whether $\phi$ is satisfiable or not). MAX-SAT is thus NP-hard as well. We emphasize that in a \msat instance, we are given a \tsat~ formula (not an \msat formula) with the task of finding the maximum number of satisfiable clauses.

\subsection{Quadratic Unconstrained Binary Optimization (QUBO)}\label{sec:PBO}
\noindent In this section we will formally introduce quadratic unconstrained binary optimization (QUBO) and related terminology that will be used in the remainder of this paper.
\begin{dfn}[QUBO~\cite{glover2018tutorial}] Let $\mathcal{Q} \in \mathbb{R}^{n \times n}$ be a square matrix and let $x \in \{0,1\}^n$ be an $n$-dimensional vector of Boolean variables. The QUBO problem is defined as follows:
\begin{equation}\label{eq:qubo}
 \text{minimize} \quad H_{\textit{QUBO}}(x) = x^T\mathcal{Q}x = 
\sum_{i}\mathcal{Q}_{ii}x_i + \sum_{i <j}\mathcal{Q}_{ij}x_ix_j
\end{equation}
\end{dfn}

 We call $H_{\textit{QUBO}}(x)$ the (QUBO) energy of vector $x$. The matrix $Q$ will also be called \emph{QUBO matrix}. Representing a QUBO matrix as an upper triangular matrix is customary. Note specifically that a QUBO matrix is just the matrix representation of the quadratic pseudo-Boolean polynomial shown in Eq. \ref{eq:qubo}. 

To transform a given \msat instance to an instance of QUBO, we will introduce additional variables that do not correspond to any variables of the given \msat instance. We will often say that an assignment $\vec{x} = (x_1 := v_0, \dots, x_n := v_n), v_i \in \{0,1\}$ of Boolean values to the variables $x_1, \dots, x_n$ of the \msat instance has energy $E$ in $Q$, by which we mean:
\begin{equation}
    min\; \{(\vec{x},y)^T Q (\vec{x},y) \;|\; y \in \{0,1\}^m\} = E
\end{equation}
Here $Q$ is a QUBO matrix and $(\vec{x},y)$ is an  $(n+m)$-dimensional column vector defined as $(\vec{x},y) = (x_1 = v_0, \dots, x_n = v_n,y_1, \dots, y_m)$. The first $n$ values of the vector $(\vec{x},y)$ are given by the assignment $\vec{x} = (x_1 := v_0, \dots, x_n := v_n), v_i \in \{0,1\}$ of Boolean values to the variables $x_1, \dots, x_n$ of the \msat instance. The last $m$ entries represent the values of the auxiliary variables $y_1, \dots, y_m$.

\subsection{\msat as QUBO: The General Idea}\label{sec:general_idea}
In this Section we first introduce the idea behind most transformations from \msat to QUBO, before explaining specific QUBO transformations used in this work.
The main observation is that each clause of a \tsat~ formula of a \msat instance contains either exactly zero, one, two, or three negations. It is well known that each of these four types of clauses can be expressed as a pseudo-Boolean function as follows \cite{verma2021efficient}:
\begin{enumerate}
    \item[1.] Zero negations $(x_i \vee x_j \vee x_k): f(x_i, x_j,x_k) = -x_i - x_j -x_k +x_{i}x_{j} +x_ix_k +x_jx_k -x_{i}x_{j}x_{k}$
    \item[2.] One negation $(x_i \vee x_j \vee -x_k): f(x_i, x_j,x_k) = -1 + x_k - x_ix_k -x_jx_k +x_ix_jx_k$
    \item[3.] Two negations $(x_i \vee -x_j \vee -x_k): f(x_i, x_j,x_k) = - 1 +x_jx_k -x_ix_jx_k$
    \item[4.] Three negations $(-x_i \vee -x_j \vee -x_k): f(x_i, x_j,x_k) = -1 +  x_ix_jx_k$
\end{enumerate}
Using this approach, whenever a clause is satisfied, $f(x_1,x_2,x_3) = -1$ holds, and in the case the clause is not satisfied, then $f(x_1,x_2,x_3) = 0$. Thus, a satisfying assignment minimizes $f(x_1,x_2,x_3)$ as desired. By applying the correct pseudo-Boolean function to all clauses of the \tsat~ formula of the \msat instance and summing all resulting pseudo-Boolean polynomials, one receives a pseudo-Boolean polynomial that is minimized by the best solution(s) to the corresponding \msat problem. \\ Note that the four cases from above create cubic terms $x_ix_jx_k$. As QUBO is concerned with minimizing a quadratic pseudo-Boolean polynomial, we must express the cubic terms as a quadratic polynomial. This can be achieved by quadratic reformulation techniques, which add an extra variable to the problem for each cubic term that gets quadratically reformulated. We refer to \cite{anthony2017quadratic} for further information on quadratic reformulation techniques. Note that the quadratic reformulation of a cubic term introduces a new auxiliary variable for each cubic term.
 \subsection{Chancellor's transformation}\label{subsec:chancellor_qubo}
The general idea of Chancellor's transformation \cite{chancellor2016direct} is to map an arbitrary clause of a satisfiability problem to an instance of QUBO such that all variable assignments that satisfy the clause have the same minimum energy in the QUBO minimization problem. Simultaneously, the single variable assignment, that does not have the minimum energy in the QUBO problem should have a higher energy. By applying QUBO mappings that follow this logic to each clause of the \tsat~ formula of a \msat instance and adding up the resulting quadratic polynomials, a QUBO $\mathcal{Q}_{instance}$ mapping for the whole \msat instance is received. As a consequence of this construction, it is guaranteed that the minimum of $\mathcal{Q}_{instance}$ corresponds to an assignment of Boolean values to the variables of the \msat instance, such that the most clauses are satisfied. After some specific logical deductions, which we cannot present in-depth due to space restrictions, Chancellor derives mappings that transform each individual clause of a \tsat~ formula to an instance of QUBO. For the four types of clauses defined in Sec. \ref{sec:general_idea} Chancellors mappings are depicted in Table \ref{tab:chancellor_qubo}.

\begin{table}[h]
   \caption{Chancellor's QUBO transformations for the four types of clauses}\label{tab:chancellor_qubo}
    \begin{minipage}{.5\linewidth}
      \caption*{(a) $( x_i \lor x_j\lor x_k )$}
      \vspace{.05cm}
\centering
\begin{tabular}[h]{|r||p{0.3cm}|p{0.3cm}|p{0.3cm}|c|}
\hline
& $x_i$ & $x_j$ & $x_k$ &$a_l$ \\
\hline
\hline
$x_i$ &-2 & 1 & 1 &1 \\
\hline
$x_j$ & & -2 & 1 &1\\
\hline
$x_k$ & & & -2 & 1 \\
\hline
$a_l$ & & & &-2 \\
\hline
\end{tabular}
    \end{minipage}%
    \begin{minipage}{.5\linewidth}
      \centering
        \caption*{(b) $( x_i \lor x_j \lor \lnot x_k )$}
        \vspace{.05cm}
        \begin{tabular}[h]{|r||p{0.3cm}|p{0.3cm}|p{0.3cm}|c|}
\hline
& $x_i$ & $x_j$ & $x_k$ & $a_l$ \\
\hline
\hline
$x_i$ & -1& 1 & & 1\\
\hline
$x_j$ & & -1& &1 \\
\hline
$x_k$ & & & & \\
\hline
$a_l$& & & &-1\\
\hline
\end{tabular}
    \end{minipage}

\vspace{.25cm}
\begin{minipage}[!h]{.5\linewidth}

      \caption*{(c) $( x_i \lor \lnot x_j \lor \lnot x_k )$}
      \vspace{.05cm}
      \centering
      \begin{tabular} [h]{|r||p{0.3cm}|p{0.3cm}|p{0.3cm}|c|}
\hline
& $x_i$ & $x_j$ & $x_k$  & $a_l$\\
\hline
\hline
$x_i$ & -1 &  &  & 1\\
\hline
$x_j$ & & -1 & 1 & 1 \\
\hline
$x_k$ & & & -1 & 1 \\
\hline
$a_l$ & & & & 2\\
\hline
\end{tabular}
    \end{minipage}%
    \begin{minipage}{.5\linewidth}
      \centering
        \caption*{(d) $( \lnot x_i \lor \lnot x_j \lor \lnot x_k )$}
        \vspace{.05cm}
      
       \begin{tabular}[h]{|r||p{0.3cm}|p{0.3cm}|p{0.3cm}|c|}
\hline
& $x_i$ & $x_j$ & $x_k$  & $a_l$ \\
\hline
\hline
$x_i$ & -1 & 1 & 1 & 1 \\
\hline
$x_j$ & & -1 & 1 & 1 \\
\hline
$x_k$ & & & -1 & 1\\
\hline
$a_l$ & & & & -1\\
\hline
\end{tabular}

    \end{minipage} 
\end{table}
The following example explains the mappings.
\begin{expl}\label{expl:chancellor}
Suppose we are given the formula $(x_1 \vee x_2 \vee x_3) \wedge (x_1 \vee x_2 \vee -x_3)$. As the first clause is of type 1 (i.e., has no negations), we apply the transformation shown in Table \ref{tab:chancellor_qubo}(a) and receive the polynomial $P_1 = -2x_1 -2x_2 -2x_3 -2_a1 +x_1x_2 + x_1x_3 + x_1a_1 + x_2x_3 +x_2a_1 + x_3a_1$. Then we apply the QUBO transformation of Table \ref{tab:chancellor_qubo}(b) to the second clause and receive polynomial $P_2 = -x_1 -x_2 - a_2 +x_1a_2 +x_2a_2$. Note that we have replaced index $l$ of variable $a_l$ with $1$ in the first clause and $2$ in the second clause. As explained in Sec. \ref{sec:general_idea} each clause needs an additional variable to represent cubic terms as a quadratic polynomial. Hence $a_l$ is the additional variable of the $l$-th clause. Summing up the polynomials $P_1$ and $P_2$ yields $P_{final} = P_1 + P_2= -3x_1 -3x_2 -2x_3 -2a_1 -a_2 +2x_1x_2 +x_1x_3 + x_1a_1 +x_1a_2 + x_2x_3 +x_2a_1 -x_2a_2 +x_3a_1 +x_3a_3$ which can be represented as a QUBO matrix as shown in Tab. \ref{tab:combined_example_qubo}.

\begin{table}[h]
   \caption{Result of combining the QUBO matrices shown in Tab. \ref{tab:chancellor_qubo}a) and Tab. \ref{tab:chancellor_qubo}b).}\label{tab:combined_example_qubo}
    
\centering
\begin{tabular}[h]{|r||p{0.3cm}|p{0.3cm}|p{0.3cm}|c|c|}
\hline
& $x_1$ & $x_2$ & $x_3$ &$a_1$ & $a_2$ \\
\hline
\hline
$x_1$ &-3 & 2 & 1 &1 &1 \\
\hline
$x_2$ & & -3 & 1 &1 &1\\
\hline
$x_3$ & & & -2 & 1 & 1\\
\hline
$a_1$ & & & &-2 & \\
\hline
$a_2$ & & & & &-1 \\
\hline
\end{tabular}

\end{table}
\end{expl} 

\subsection{Nüßlein's transformation}
Nüßlein's transformation \cite{nusslein2023solving} employs a similar idea as Chancellor's transformation. Each clause is mapped to an instance of QUBO $\mathcal{Q}$ such that all satisfying assignments of a clause have the same minimal energy in $\mathcal{Q}$. In contrast, the single non-satisfying assignment of the clause has a higher, non-optimal energy. Nüßlein observed that the mapping of the clauses to instances of QUBO could be realized by the mappings depicted in Tab. \ref{table:qubo_nuesslein}.

\begin{table}[!htb]
   \caption{Nüßlein's QUBO transformations for the four different types of clauses \cite{nusslein2023solving}}
    \begin{minipage}{.5\linewidth}
      \caption*{(a) $( a \lor b\lor c )$}
       \vspace{.05cm}
\centering
\begin{tabular}[h]{|r||p{0.3cm}|p{0.3cm}|p{0.3cm}|c|}
\hline
& a & b & c & K \\
\hline
\hline
a & & 2 & & -2 \\
\hline
b & & & & -2 \\
\hline
c & & & -1 & 1 \\
\hline
K & & & & 1 \\
\hline
\end{tabular}
    \end{minipage}%
    \begin{minipage}{.5\linewidth}
      \centering
        \caption*{(b) $( a \lor b \lor \lnot c )$}
         \vspace{.05cm}
        \begin{tabular}[h]{|r||p{0.3cm}|p{0.3cm}|p{0.3cm}|c|}
\hline
& a & b & c & K \\
\hline
\hline
a & & 2 & & -2 \\
\hline
b & & & & -2 \\
\hline
c & & & 1 & -1 \\
\hline
K & & & & 2 \\
\hline
\end{tabular}
\label{table:qubo_nuesslein}
    \end{minipage}

\vspace{.25cm}
\begin{minipage}[!h]{.5\linewidth}

      \caption*{(c) $( a \lor \lnot b \lor \lnot c )$}
       \vspace{.05cm}
      \centering
      \begin{tabular} [h]{|r||p{0.3cm}|p{0.3cm}|p{0.3cm}|c|}
\hline
& a & b & c & K\\
\hline
\hline
a & 2 & -2 & & -2 \\
\hline
b & & & & 2 \\
\hline
c & & & 1 & -1 \\
\hline
K & & & & \\
\hline
\end{tabular}
    \end{minipage}%
    \begin{minipage}{.5\linewidth}
      \centering
        \caption*{(d) $( \lnot a \lor \lnot b \lor \lnot c )$}
         \vspace{.05cm}
       \begin{tabular}[h]{|r||p{0.3cm}|p{0.3cm}|p{0.3cm}|c|}
\hline
& a & b & c & K \\
\hline
\hline
a & -1 & 1 & 1 & 1 \\
\hline
b & & -1 & 1 & 1 \\
\hline
c & & & -1 & 1 \\
\hline
 K & & & & -1 \\
\hline
\end{tabular}
    \end{minipage} 
     
\end{table}
The QUBO transformations for each of the clause types are applied to the clauses of an \msat problem in the same way as demonstrated in Example \ref{expl:chancellor}.

\subsection{\pq method}
As explained in the previous sections, to transform a given \msat instance to an instance of QUBO, it suffices to transform each clause of the \tsat{} formula associated with the \msat problem to an instance of QUBO. All the QUBO instances resulting from the transformation of the clauses will then be combined into a single QUBO instance (as demonstrated in Sec. \ref{subsec:chancellor_qubo}). The \pq method \cite{zielinski23Pattern} is a meta-approach that can automatically identify all possible transformations from a given clause to an instance of QUBO. As explained in Sec. \ref{sec:general_idea} to transform a clause to an instance of QUBO, an additional variable is needed. Thus, the \pq method is given a blank $(4\times4)$-dimensional upper triangular QUBO matrix as input. The user then specifies a set of values that the \pq method can use to insert into the blank QUBO matrix. Given this set of values and a specific clause, the \pq method exhaustively searches the space of $4\times4$-dimensional upper triangular QUBO matrices that only consist of values of the user-specified set of values. The goal of this exhaustive search procedure is to find QUBO matrices in which all the minima correspond to satisfying solutions to the given clause. This way, the \pq method finds all possible methods of transforming a given clause to an instance of QUBO. Similarly to the previously explained methods by Chancellor and Nüßlein, by applying the respective transformations to the clauses of the \tsat{} formula and combining the resulting QUBO matrices, we receive a QUBO representation of the \msat instance.

%% file: chapters/3-RelatedWork.tex
\section{Related Work}\label{sec:related_work}
\noindent In a paper by Sax et al. \cite{sax2020approximate}, the authors studied the creation of QUBO approximations for several problem classes by randomly removing entries from a given initial QUBO matrix. By removing values from a QUBO matrix, fewer physical qubits are needed to solve the QUBO on a quantum annealer. This study thus researched how much the problem size (resp. the number of needed physical qubits) could be reduced without worsening the solution quality too much. For the class of \tsat{} problems, the authors transformed each \tsat{} instance to an instance of QUBO using Choi's method \cite{choi2010adiabatic}. The authors continuously removed values from the initial QUBO created by Choi's method until only diagonal entries were left in the QUBO matrix. They observed that up to 70\%  of the initial QUBO entries could be removed without observing a significant decline in the solution quality.
Apart from this study, QUBO approximation has mostly been studied in contexts where an objective function that is to be expressed as a QUBO instance is unknown or when objectives or constraints of a problem cannot be encoded effectively in a quadratic model. In \cite{nusslein2023black}, the authors present a method to approximate the low-energy spectrum of a problem as a QUBO using a black-box approach. This method involves concurrently conducting a regression on the low-energy spectrum and differentiating between high-energy and low-energy states through classification. They demonstrate that this approach yields a significantly more accurate approximation of the low-energy spectrum, resulting in enhanced optimization performance compared to merely performing regression on all sampled states.
\noindent In a study by Matsumori et al. \cite{matsumori2022application}, the authors are concerned with the creation of a QUBO approximation of a design optimization problem, because the objectives and constraints of the design optimization problem cannot be expressed explicitly as a QUBO problem. In their approach, a black-box optimization approach based on the factorization machine \cite{rendle2010factorization} is used to create a QUBO approximation of the input problem. 

%% file: chapters/4-ApproximationMethods.tex
\section{\msat QUBO Approximation Methods}\label{sec:app_methods}
\noindent In this section, we will define the concept of QUBO approximation and detail our approach to systematically creating such QUBO approximations for \msat instances. We will evaluate these approaches in Sec. \ref{sec:eval}.

\begin{dfn}[QUBO Approximation for MAX-3SAT] Let $\psi$ be a \msat instance. A QUBO instance $\mathcal{Q}$ is a \emph{QUBO approximation} of $\psi$ if there are optimal solutions of $\psi$ that do not have minimal energy in $Q$ (Sec. \ref{sec:PBO} defines how to calculate the energy of an assignment). 
\end{dfn}
\noindent Intuitively, this definition states that some of the optimal solutions of a given \msat instance cannot be found by minimizing a given QUBO approximation for this \msat instance. \\
In the following sections, we will introduce two strategies for creating QUBO approximations of \msat problems. In Sec. \ref{subsec:naive_qubo_approx}, we will define naive QUBO approximation analogously to Sax et al. \cite{sax2020approximate}. In Sec. \ref{subsec:sys_qubo_approx}, we present the core contribution of this paper: our method of systematically creating QUBO approximations of \msat instances.

\subsection{Naive QUBO Approximation}\label{subsec:naive_qubo_approx}
\noindent In their study, Sax et al. \cite{sax2020approximate} created QUBO approximations by removing values from QUBO matrices that resulted from applying Choi's transformation to a set of \tsat{} formulas. In this paper, we will use the term \emph{naive QUBO approximation} for creating QUBO approximations for \msat problems by removing values from a given, non-approximated QUBO representation. Thus, Sax et al. conducted naive QUBO approximation for Choi's method. Choi's method leads to QUBO matrices of dimension $3m \times 3m$, where $m$ is the number of clauses of a \tsat{} formula. In this paper, however, we want to begin our study of QUBO approximation methods for \msat problems by applying naive QUBO approximation to a class of QUBO transformations that results in $(n+m) \times (n+m)$-dimensional QUBO matrices, where $n$ is the number of variables and $m$ is the number of clauses of a \tsat{} problem. This class contains many thousand different QUBO transformations \cite{zielinski23Pattern}. Thus, we define two methods of creating naive QUBO approximations analogously to Sax et al. \cite{sax2020approximate}:\\
\begin{enumerate}
    \item[1.]\textbf{Min Pruning}: We are given an initial, exact, non-approximated QUBO representation of an \msat instance, consisting of the QUBO matrix $\mathcal{Q}$. To create a naive QUBO approximation, we proceed to remove the $N$ smallest entries of $\mathcal{Q}$. 

    \item[2]\textbf{Random Pruning}: Random Pruning generally works similarly, except that we now remove $N$ randomly chosen values, no matter their sign or magnitude.
\end{enumerate}

\subsection{Systematic QUBO Approximation}\label{subsec:sys_qubo_approx}
\noindent In this section, we present the core contribution of our paper: We introduce a new and systematic approach to approximating QUBO representations of \msat instances. As explained in Sec. \ref{sec:general_idea}, the general idea of $(n+m)\times(n+m)$-dimensional QUBO transformations for \msat instances is to transform each clause of the \tsat~ formula of an \msat instance into an instance of QUBO and sum up all the resulting QUBO representations of the clauses. This way, we receive a QUBO in which all minima correspond to optimal solutions of the \msat instance. To accurately transform a \tsat~ clause into an instance of QUBO, the QUBO instance needs to contain an additional auxiliary variable (see Sec. \ref{sec:PBO}). Thus each QUBO matrix that results from transforming a \tsat~ clause into an instance of QUBO is of dimension $(4\times4)$ (three variables correspond to the variables of the clause and one auxiliary variable). The idea of our approach to systematically approximate QUBO transformations for \msat problems is to create $(n \times n)$-dimensional QUBO representations instead of $(n+m) \times (n+m)$-dimensional ones. In particular, we transform each clause to an instance of QUBO with dimension $(3 \times 3)$ instead of $(4\times4)$. Hence, our approach does not contain any auxiliary variables. The following theorem shows that transforming \tsat~ clauses into instances of QUBO, which consist of $(3\times3)$-dimensional QUBO matrices, is necessarily a QUBO approximation.

\begin{theorem}\label{theorem:approx}
Let $x_1, x_2, x_3$ be the variables of a clause of a \msat instance. Let $S_{SAT}$ be the set of all assignments of Boolean values to $x_1, x_2, x_3$ that satisfy the clause and let $S_{UNSAT}$ be the set of all assignments of Boolean values to $x_1,x_2,x_3$ that do not satisfy the clause. Thus $S_{SAT}$ and $S_{UNSAT}$ are sets consisting of 3-tuples ($k_1$, $k_2$, $k_3$), $k_1, k_2, k_3$ $\in \{0,1\}$, where $k_i$ denotes the value of $x_i$ for $i \in \{1,3\}$.  Let $f(x_1,x_2,x_3):= \alpha_1x_1 + \alpha_2x_2 + \alpha_3x_3 + \alpha_{12}x_1x_2 + \alpha_{13}x_1x_3 + \alpha_{23}x_2x_3$ be a polynomial in $x_1,x_2,x_3$ with $\alpha_i \in \mathbb{R}$. Then there do not exist choices of $\alpha_1, \alpha_2, \alpha_3, \alpha_{12}, \alpha_{13}, \alpha_{23}, E \in \mathbb{R}$ such that $f(s) = -E$ and $f(u) > -E$ for each $s \in S_{SAT}$ and for each $u \in S_{UNSAT}$.
\end{theorem}
\begin{proof}
We have to prove Theorem \ref{theorem:approx} for all four types of clauses (see Sec. \ref{sec:general_idea}).
Without loss of generality, we proof Theorem \ref{theorem:approx}  for clauses of type 1 like $(x_i \vee x_j \vee x_k)$. We will use a bitstring of length $3$ like \textbf{100} to denote the assignment $x_i = 1, x_j = 0, x_k = 0$.\\
Assume to the contrary, that we can choose  $\alpha_1, \alpha_2, \alpha_3, \alpha_{12}, \alpha_{13}, \alpha_{23}, E \in \mathbb{R}$ such that $f(s) = -E$ and $f(u) > -E$ for each $s \in S_{SAT}$ and for each $u \in S_{UNSAT}$. As \textbf{100}, \textbf{010} and \textbf{001} are all satisfying assignments for clauses of type 1, $f(100) = f(010) = f(001) = -E$. Thus, $\alpha_1 = \alpha_2 = \alpha_3 = -E$. Since \textbf{110} is a satisfying assignment, it follows that  $f(110) = \alpha_1x_i + \alpha_2x_j + \alpha_{12}x_ix_j = -E$. Since $\alpha_1 = \alpha_2 = -E$, we conclude that $\alpha_{12} = E$. Similarly, we show that $\alpha_{13} = \alpha_{23} = E$. But then $f(111) = -E -E -E +E +E +E = 0$, which is a contradiction, because \textbf{111} is a satisfying assignment and thus it should hold that $f(111) = -E$. 
\end{proof}\ 
As a \msat clause consists of three Boolean variables, there are $2^3 =8$  possible assignments of Boolean values to the variables of the clause. For every clause, exactly seven of these assignments satisfy it, and exactly one assignment does not. Theorem \ref{theorem:approx} shows that it is not possible to encode all 7 satisfying solutions of a clause as the minima of a QUBO minimization problem, consisting of a $(3\times3)$-dimensional QUBO matrix $\mathcal{Q}$. However, the proof of Theorem \ref{theorem:approx} shows that it is possible to encode 6 of the 7 satisfying solutions of a clause as a minimum of a QUBO minimization problem. This means that one satisfying solution has a non-optimal energy. By using this approximation approach, we save one auxiliary variable per clause of the \tsat~ formula of the \msat instance, as well as some quadratic coefficients we would have otherwise had to add in an exact QUBO transformation of dimension $(n+m) \times (n+m)$. Consequently, fewer physical qubits are needed to solve these QUBO approximations on quantum hardware. In return for this gain, we can no longer guarantee that every optimal solution to the \msat problem also has minimal energy in the approximated QUBO minimization problem.
 We will demonstrate that this trade-off is worthwhile in Sec. \ref{subsec:eval_systematic}. \\
We will now show how to systematically find $(3\times 3)$-dimensional QUBOs, for which six of the seven satisfying assignments of a clause have the same energy, while the remaining two assignments have a higher energy in $\mathcal{Q}$. Our approach adapts the \pq method \cite{zielinski23Pattern}. Instead of searching within the space of $(4\times 4)$-dimensional QUBO matrices and looking for those where all satisfying assignments have the same minimal energy, we are looking for $(3\times3)$-dimensional QUBO matrices in which six of the seven satisfying assignments of a clause have the same minimal energy. Thus, our task is to assign values to the variables $\alpha_{1}, \dots, \alpha_{23}$ of the prototype of a $(3\times 3)$-dimensional QUBO matrix shown in Table \ref{tab:n_qubo_prototype}, such that the aforementioned condition is satisfied. 

\begin{table}
\centering
\caption{Prototype of a QUBO approximation for a \msat Clause}\label{tab:n_qubo_prototype}
\begin{tabular}[h]{|r||p{0.3cm}|p{0.3cm}|p{0.3cm}|c|}
\hline
& $x_i$ & $x_j$ & $x_k$ \\
\hline
\hline
$x_i$ & $\alpha_1$& $\alpha_{12}$ &$\alpha_{13}$  \\
\hline
$x_j$ & & $\alpha_2$ & $\alpha_{23}$ \\
\hline
$x_k$ & & & $\alpha_ 3$ \\
\hline
\end{tabular}
\end{table}

\noindent Our algorithm to find $(3\times3)$-dimensional QUBO approximations of \tsat~ clauses is as follows:
\begin{algorithm}
\begin{algorithmic}[1]\label{algorithm_code}
\Require Set $S$ of values the search method can insert into the QUBO matrix as values for $\alpha_1, \alpha_2, \alpha_3, \alpha_{12}, \alpha_{13},\alpha_{23}$.
\Require Set $A = \{(\alpha_1, \alpha_2, \alpha_3, \alpha_{12}, \alpha_{13}, \alpha_{23}) \in $S$^6$\} containing all possible six tuples of values of $S$.
\Procedure{Search QUBO Approximation}{clause\_type}
\State FoundQUBOS = \{\}
\For{tuple in $A$}:
\If{six satisfying assignments for clause of type \emph{clause\_type} have minimal energy}
\State FoundQUBOs.insert(tuple)
\EndIf
\EndFor
\State \textbf{return} FoundQUBOS
\EndProcedure
\end{algorithmic}
\end{algorithm}
The input to the algorithm is a set $S$ of values that can be assigned to $\alpha_1, \dots, \alpha_{23}$. In this paper, we will use $S:=\{-1,0,1\}$.
Then, we specify for which type of clause (see Sec. \ref{sec:general_idea}) the search method should find QUBO approximations. This is needed because we need six of the seven satisfying assignments of a clause to have the same optimal energy in the resulting QUBO problem. As different types of clauses are satisfied by a different set of assignments, the algorithm needs the respective clause type as an input. The method then performs an exhaustive search, trying all possible combinations of assignments of values of $S$ to variables $\alpha_1, \dots, \alpha_{23}$, to find QUBO matrices, for which six of the seven possible satisfying assignments of the specified type of clause have minimal energy. Performing this procedure for all four types of clauses, we receive 4 QUBO approximations for each of the four types of clauses within a few seconds.\\
Remember, as QUBO is just the matrix representation of a quadratic polynomial, the found QUBO approximations for any of the four types of clauses each define a mapping from a \tsat~ clause type to a quadratic polynomial. Fixing one QUBO approximation per clause type thus yields a method of transforming each clause of a \tsat~ formula of a \msat instance into a quadratic polynomial. Summing up all the polynomials that result from applying  QUBO approximations of the correct clause type to the clauses of the \tsat~ formula of the \msat instances yields a QUBO approximation of the whole \msat instance (see example in Sec. \ref{subsec:chancellor_qubo}). As the search resulted in four QUBO approximations for each clause type, there are $256 = 4 \cdot 4 \cdot 4 \cdot 4$ QUBO approximations for an \msat instance. To choose one of these 256 QUBO approximations for an evaluation (see Sec. \ref{sec:eval}), we created a single \tsat{} formula, according to the method described in \ref{subsec:eval_dataset}. We then transformed used each of the 256 QUBO approximations to create a QUBO instance corresponding to the \tsat{} formula. We solved all 256 resulting QUBO instances with D-Wave's tabu search and chose the best-performing approximation (i.e., the one that satisfied most clauses), which is shown in Tab. \ref{tab:full_approx_qubo}, for practical evaluation.

\begin{table}
   \caption{Approximated QUBOs for the four different types of clauses}\label{tab:full_approx_qubo}
    \begin{minipage}{.5\linewidth}
      \caption*{(a) $( x_i \lor x_j\lor x_k )$}
      \vspace{.05cm}
\centering
\begin{tabular}[h]{|r||p{0.3cm}|p{0.3cm}|p{0.3cm}|c|}
\hline
& $x_i$ & $x_j$ & $x_k$ \\
\hline
\hline
$x_i$ &-1 & 1 & 1 \\
\hline
$x_j$ & & -1 & 1\\
\hline
$x_k$ & & & -1 \\
\hline
\end{tabular}
    \end{minipage}%
    \begin{minipage}{.5\linewidth}
      \centering
        \caption*{(b) $( x_i \lor x_j \lor \lnot x_k )$}
        \vspace{.05cm}
        \begin{tabular}[h]{|r||p{0.3cm}|p{0.3cm}|p{0.3cm}|c|}
\hline
& $x_i$ & $x_j$ & $x_k$  \\
\hline
\hline
$x_i$ & & 1 &-1\\
\hline
$x_j$ & & & -1 \\
\hline
$x_k$ & & & 1 \\

\hline
\end{tabular}
    \end{minipage}

\vspace{.25cm}
\begin{minipage}[!h]{.5\linewidth}

      \caption*{(c) $( x_i \lor \lnot x_j \lor \lnot x_k )$}
      \vspace{.05cm}
      \centering
      \begin{tabular} [h]{|r||p{0.3cm}|p{0.3cm}|p{0.3cm}|c|}
\hline
& $x_i$ & $x_j$ & $x_k$ \\
\hline
\hline
$x_i$ & 1 & -1 & -1\\
\hline
$x_j$ & & & 1 \\
\hline
$x_k$ & & & \\
\hline
\end{tabular}
    \end{minipage}%
    \begin{minipage}{.5\linewidth}
      \centering
        \caption*{(d) $( \lnot x_i \lor \lnot x_j \lor \lnot x_k )$}
        \vspace{.05cm}
      
       \begin{tabular}[h]{|r||p{0.3cm}|p{0.3cm}|p{0.3cm}|c|}
\hline
& $x_i$ & $x_j$ & $x_k$    \\
\hline
\hline
$x_i$ & -1 & 1 & 1 \\
\hline
$x_j$ & & -1 & 1  \\
\hline
$x_k$ & & & -1 \\

\hline
\end{tabular}

    \end{minipage}

\end{table}

\ \\
As explained previously, our systematic QUBO approximation approach cannot guarantee that every solution of a \msat problem has minimal energy in the resulting QUBO. Interestingly, despite being an approximation, it is still possible that some optimal solutions of the \msat instance have minimal energy in the QUBO approximation. We will show this by proving the following theorem.
\begin{theorem}\label{theorem:approx}
For any \msat instance $\psi$ consisting of $n$ variables and $m$ clauses, there is at least one QUBO approximation $\mathcal{Q}_{opt}$, that yields the best possible solution to the given \msat when minimized. 
\end{theorem}
By QUBO approximation, we refer to the specific QUBO approximation we introduced in this secion.
\begin{proof}
Suppose   $as_{opt} := (x_1 = as_1, x_2 = as_2, \dots x_n = as_n), as_i \in \{0,1\}$ for $ 1 \leq i \leq n$ is an optimal assignment of Boolean values to the variables of $\psi$ such that no assignment satisfies more clauses of $\psi$ than $as_{opt}$. As explained in Sec. \ref{sec:general_idea} and Sec. \ref{subsec:chancellor_qubo}, to create a QUBO representation of $\psi$,  it suffices to create QUBO representations of each individual clause of $\psi$ and summing up the resulting quadratic polynomials. To create a QUBO representation for which $as_{opt}$ has minimum energy, we thus only need to guarantee that for each clause of $\psi$, the assignment of Boolean values to the variables of the clause, which is given by $as_{opt}$, minimizes the QUBO representation of the clause, if it satisfies the clause. As $as_{opt}$ then minimizes each QUBO representation of each individual clause of $\psi$ $as_{opt}$ satisfies, it follows that $as_{opt}$ minimizes the sum of the QUBO representations (quadratic polynomials) of the clauses, which is the QUBO representation of $\psi$. Thus, to prove Theorem \ref{theorem:approx}, we only need to show that, for each satisfying assignment of a clause, there is at least one QUBO approximation of that clause, in which the respective assignment has minimal energy. Let $S_{SAT} = \{s_1, s_2, \dots, s_7\}$ be the set of the 7 satisfying assignments for a given clause. As stated in this section, the search method finds four approximations for each clause type that assign minimum energy to six of the seven satisfying solutions for the given clause. We choose one of these four QUBO approximations arbitrarily and call it $Q_{a_1}$. Assume w.o.l.g that $s_1$, ..., $s_6$ have minimal energy in a $Q_{a_1}$ All that is left to do is analyze whether $s_7$ has minimal energy in any of the remaining three QUBO approximations the search method has found for this clause. It turns out, that 
there is always at least one QUBO approximation, say $Q_{a_2}$ amongst the four QUBO approximations the search method found for the given clause, such that $s_7$ has minimal energy in $Q_{a_2}$. Thus, when the optimal assignment for the given class (given by $as_{opt}$) is an element of $\{s_1, s_2, \dots, s_6\}$, we choose $Q_{a_1}$ as the QUBO approximation of the clause, else we choose $Q_{a_2}$. In either case, the assignment for a clause (defined by $as_{opt}$) minimizes the QUBO approximation for the clause.
\end{proof}

%% file: chapters/5-Evaluation.tex
\section{Evaluation}\label{sec:eval}
This section provides an empirical evaluation of both approximation approaches detailed in Sec. \ref{sec:app_methods}.
\subsection{Dataset}\label{subsec:eval_dataset}
To evaluate the efficacy of naive QUBO approximation for $(n+m)\times (n+m)$-dimensional QUBO transformations as well as our systematic approach of creating $(n\times n)$-dimensional QUBO approximations of \msat instances, we created a dataset of 100 \tsat~formulas. We determined experimentally that the \tsat~ formulas should not possess more than 500 clauses, as otherwise $(n+m)\times (n+m)$-dimensional QUBO instances that resulted from transforming the \tsat~ formulas could no longer be embedded onto D-Wave's QPU architecture. All 100 formulas were created randomly using the Balanced SAT method \cite{spence2017balanced}. We chose this method because it creates \tsat{} formulas for which no approach of exploiting their structure is known \cite{spence2017balanced}. The authors of \cite{spence2017balanced} empirically determined that instances generated by this method are potentially hard to solve for \sota~SAT solvers if they possess approximately 3.6 times more clauses than variables. Using a SAT solver, we observed that all formulas we created, consisting of 500 clauses and 3.6 times more clauses than variables, were unsatisfiable. Thus, all of our 100 \tsat{} formulas consist of 500 clauses and 145 variables ($\approx$ 3.45 times more clauses than variables). Note that this slight deviation of the ratio of the number of clauses and the number of variables does not compromise the difficulty of the formulas. 

\subsection{Evaluation of the Naive QUBO Approximation Methods}
\noindent As described in Sec. \ref{subsec:naive_qubo_approx}, the naive QUBO approximation methods work by randomly removing values from an exact, non-approximated QUBO transformation method. In this paper, we will use Chancellor's and Nüßlein's transformations as the exact QUBO transformations. For each of the 100 formulas in our dataset, and for each of the transformation methods (Chancellor and Nüßlein), we perform the following steps:
\begin{enumerate}
    \item[1.] Transform the given \tsat{} formula to an instance of QUBO $Q_{initial}$.
    \item[2.]Calculate the initial number of non-zero, non-diagonal values $N_{initial}$. Calculate $N_{10}:= 0.1 \times N_{initial}$.
    \item[3.]
    \begin{enumerate}
        \item Min Pruning: Create the QUBO matrix $Q_{10}$ by removing the $N_{10}$ smallest values from $Q_{initial}$. Continue to create $Q_{20}$ by removing  $N_{10}$ values from $Q_{10}$. This procedure is repeated until $Q_{100}$ is created, which is a QUBO matrix that only contains the main diagonal of $Q_{initial}$ and no non-diagonal entries.
        \item Random Pruning: Generally the same as Min Pruning, except that not the $N_{10}$ smallest entries are removed, but $N_{10}$ randomly chosen non-zero, non-diagonal entries.
    \end{enumerate}
\end{enumerate}
Thus, for each formula and for each of the QUBO transformations we receive 11 QUBO representations ($Q_{initial}$, $Q_{10}$, $\dots$, $Q_{100}$). Each of these QUBO matrices is solved 1000 times on D-Wave's quantum annealer Advantage\_System6.4.

From these 1000 samples we generated for each formula, each QUBO transformation (Chancellor, Nüßlein), each pruning method (min pruning, random pruning) and each pruning percentage (0\%, 10\%, ..., 100\% non-zero, non-diagonal values) we select the best answer, i.e., the answer that satisfied the highest number of clauses. We then calculate the average of the best answers for each (pruned) QUBO transformation. The results of this evaluation are shown in Fig. \ref{fig:min_random_remove_eval}.
\begin{figure}[h]
\centering
\includegraphics[width=.49\textwidth]{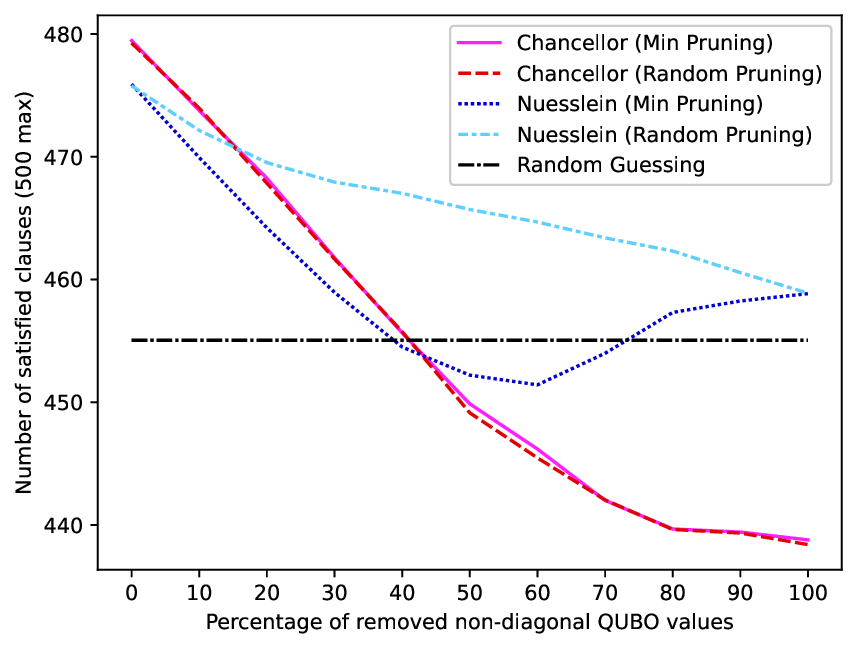}
\caption{Average of the best solutions found for each combination of pruning strategy, QUBO transformation and prune percentage.}\label{fig:min_random_remove_eval}
\end{figure}
It can be seen that no method has found a satisfying solution for any formula, as no result yielded 500 satisfied clauses. Observe, that randomly guessing solutions on average satisfied 455 out of 500 clauses. It is well known, that randomly guessing solutions for a \msat instance can satisfy approximately 7/8 of all clauses of the instance \cite{haastad2001some}. Thus, we expected random guessing to find assignments that satisfy approximately 438 of our 500-clause problems. The fact that in our case random guessing solved more than the expected 438 clauses can be attributed to the small size of each formula (500 clauses). Regardless of the used pruning method and regardless of the initial non-approximated QUBO transformation method (Chancellor or Nüßlein), we observe that removing values from the initial QUBO matrix leads to an immediate decline in the solution quality, i.e., in fewer clauses being solved. Additionally, we observe that the two QUBO formulations seem to behave differently when values are removed. In the case of Chancellor's transformation, there is a straight decline in solution quality, even below the success rate of randomly guessing solutions. In the case of Nüßlein's transformation, we can see that for both, min pruning and random pruning, the decline in the quality of  solutions (i.e., the number of clauses satisfied) stops at some point. Furthermore, we can also see for Nüßlein's transformation that using min pruning reduces the quality of the solutions faster compared to random pruning. As both pruning strategies ultimately lead to the same QUBO representation (i.e., $Q_{100}$ contains only the initial main diagonal in any case), we can see that the final solution quality is identical. Although it would be interesting to identify the reason why in Chancellor's case pruning leads to a constant decline in solution quality while the decline in the solution quality of Nüßlein's pruned QUBO  representations seems to stop at some point, as well as the difference of random and min pruning in Nüßlein's case, it is beyond the scope of this paper. We leave this analysis for future research. In this work, we are only interested in the observation that the solution quality immediately declines when values are removed. This shows that the results of the study conducted in \cite{sax2020approximate}, where 70\% of the non-zero quadratic QUBO values could be removed without significant loss in solution quality, only holds for Choi's transformation and is not in general a good strategy for creating QUBO approximations. In contrast to the results of these naive approximation methods, we will demonstrate in the following section that our approach of systematically creating QUBO approximations can even increase the solution quality.

\subsection{Evaluation of the Systematic QUBO Approximation method}\label{subsec:eval_systematic}
\noindent In this section, we evaluate our proposed method of systematically creating QUBO approximations, as described in Sec. \ref{subsec:sys_qubo_approx}. We transformed all \tsat{} instances to instances of QUBO using the QUBO approximation shown in Table \ref{tab:full_approx_qubo} (which we call \emph{FullApprox}), Chancellor's transformation, and Nüßlein's transformation. As embedding a QUBO onto D-Wave's QPU is a heuristic process, we generated 10 embeddings for each formula to reduce the influence of particular embeddings (for example different embeddings may require a different number of physical qubits). For each embedding, we then generated 100 samples. Thus, for each formula and each QUBO transformation, we generated 1000 samples, yielding 100,000 D-Wave samples per QUBO transformation. Furthermore, for each formula, we randomly guessed 1000 solutions to compare the D-Wave results against a random guessing baseline. For each formula and each method of solving the \msat problem, we used the best of the 1000 received answers to compare with the results of the other methods. The distribution of the best results for each formula and each method is shown in Fig \ref{fig:full_approx_eval}a. 
To determine the relative performance of the QUBO transformation approaches and the random guessing method, we compared the best solution of the FullApprox method with the best solution of all the other methods. That is, for each formula, we calculated the number of satisfied clauses by the best solution of the FullApprox method and subtracted the number of clauses satisfied by the best solution of any other method. Similarly, we compared the best solution for Chancellor's transformation (resp. Nüßlein's transformation) to the best solution for the random method for each formula. The results of this comparison are shown in Fig \ref{fig:full_approx_eval}b). A label  \emph{A,B} on the x-axis denotes that the respective boxplot shows the results of comparing method A against method B. That is,  label \emph{FA, C} denotes that the results of the FullApprox transformation are compared to Chancellor's transformation as described above.

\begin{figure}[h]
\centering
\includegraphics[width=.49\textwidth]{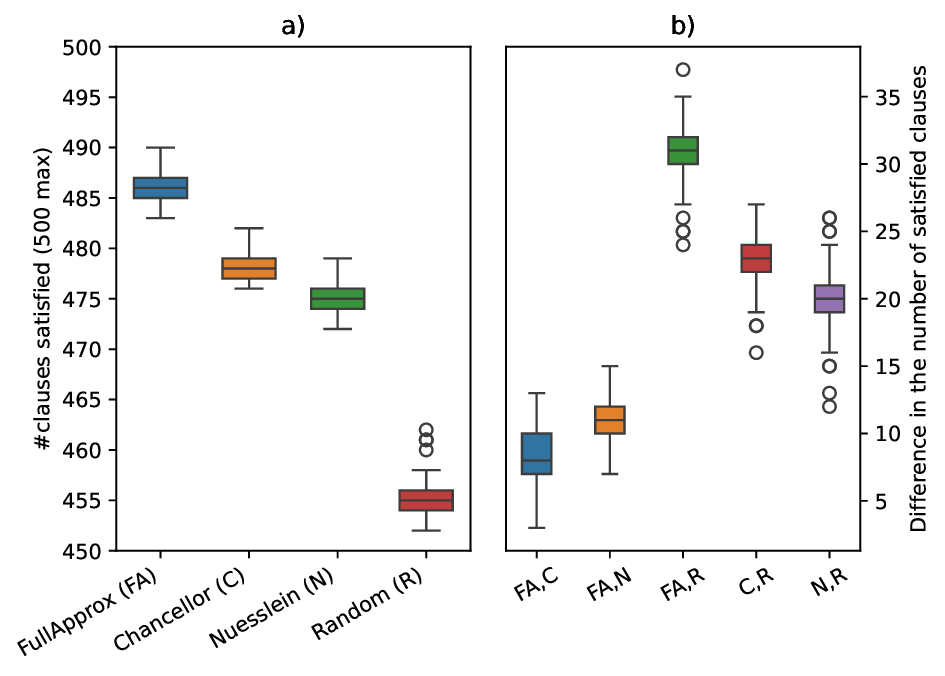}
\caption{Distribution of the best solutions found with different QUBO transformations (left). Difference of the best solutions found by one approach compared to the best solutions found by another approach (right).}\label{fig:full_approx_eval}
\end{figure}
From Fig. \ref{fig:full_approx_eval}a we can see that no satisfying solution was found for either of the approaches since no answer of any approach yielded 500 satisfied clauses. However, as we are solving MAX-SAT, we are interested in the approach that satisfies the largest number of clauses. The best results, i.e., the most satisfied clauses, were obtained using our FullApprox method. The exact QUBO transformations by Chancellor and Nüßlein and the random method yielded fewer satisfied clauses for each formula. The direct comparison (Fig. \ref{fig:full_approx_eval}b) shows that the best solution for all the formulas was found by the FullApprox method, as Chancellor's transformation yielded between 3 and 13 less satisfied clauses,  Nüßlein's transformation between 7 and 15 less satisfied clauses, and the random baseline  between 24 and 37 less satisfied clauses for each formula compared to the FullApprox method. As explained in Sec. \ref{fig:full_approx_eval}, the results of random guessing are as expected, as it is possible to satisfy approximately 7/8 of all clauses by randomly guessing \cite{haastad2001some}. This emphasizes an important point about the interpretation of the results: comparing the absolute number of satisfied clauses of solutions generated by different methods does not yield a good indication of the performance of different methods. For example, a perfect solver will find a solution that satisfies all the clauses, while randomly guessing satisfies  \emph{only} 12.5\% (=1/8) fewer clauses. We thus argue that the number of clauses satisfied by the random guessing method should be the baseline for each formula.
Using the number of satisfied clauses found by the random guessing method for each formula as the baseline, detailed data analysis showed that the FullApprox method yields between 12\% and 59\%  more satisfied clauses than Chancellor's method. A similar analysis showed that the FullApprox method yields between 28\% and 125\% more satisfied clauses than Nüßlein's method.

As the FullApprox approach is an approximation, it needs a significantly smaller amount of physical qubits on a quantum computer than the exact QUBO transformations by Chancellor and Nüßlein, to solve the same \msat problems. Table \ref{tab:embedding_sizes} shows the embedding sizes (i.e., the number of physical qubits needed) on D-Wave's Quantum Annealer \emph{Advantage\_System6.4} for the formulas and QUBO transformations of the previously conducted evaluation.

\begin{table}[!t]
\renewcommand{\arraystretch}{1.3}
\caption{Embedding sizes (i.e., number of physical qubits) on D-Wave's Advantage\_System6.4 (5,614 available qubits)}\label{tab:embedding_sizes}
\label{table_example}
\centering
\begin{tabular}{l||c|c|c}
\hline
\bfseries QUBO Transformation & \bfseries Min Size & \bfseries Max Size & Avg Size\\
\hline\hline
FullApprox & 2,203 & 2,819 & 2,385\\
\hline
Chancellor & 4,653 & 4,937 & 4,769 \\
\hline
Nüßlein & 4,640 &4,911 & 4,763\\
\end{tabular}
\end{table}
It is notable that our QUBO approximation approach \emph{FullApprox} only needed half of the number of physical qubits on D-Wave's Advantage\_System6.4 than the other approaches, while at the same time yielding better results (as shown in Fig. \ref{fig:full_approx_eval}). Note that the Advantage\_System6.4 possesses 5,612 physical qubits. Considering the embedding sizes in Table \ref{tab:embedding_sizes}, it is clear that when using exact approaches (like Chancellor's and Nüßlein's transformations), one cannot solve formulas with much more than 500 clauses on this machine. However, using the approximation approach, we have been able to embed formulas with 750 clauses on the Advantage\_System6.4, which used approximately 4,500 qubits. Thus, by using our systematic QUBO approximation method, one can embed 50\% larger formulas on the quantum annealer compared to non-approximated $(n+m)\times(n+m)$-dimensional QUBO transformations (like Chancellor's and Nüßlein's method).

\subsection{Scaling of the Systematic QUBO Approximation Method}
\noindent As seen in the previous section, our systematic QUBO approximation method can yield even better results than exact QUBO transformations. In this section, we want to address the scaling behavior of this approach. We will show that this approach does not only work when the formulas are small but can still yield good solutions as the size of the \msat problems grows.\\ 
For this section, we used the Balanced SAT \cite{spence2017balanced} method again to create 100 \tsat{} formulas each consisting of 10,000 clauses and 2,780 variables, which amounts to a ratio of clauses to variables of approximately 3.6. This is the empirically derived clauses-to-variable ratio for hard-to-solve \tsat{} instances of the Balanced SAT method. As hard formulas of this size can take multiple days (or even longer) to get solved, we do not know whether these formulas are satisfiable (i.e., the maximum number of satisfiable clauses is 10,000) or not (i.e., the maximum number of satisfiable clauses is smaller than 10,000). As we are only interested in the relative results between different approaches, this is not a limitation.\\
We transformed each of these 100 formulas to instances of QUBO using the previously described FullApprox method and Chancellor's and Nüßlein's transformation. Each QUBO will be solved on a classical (i.e., non-quantum) computer 100 times using D-Wave's tabu search implementation. Thus, for each formula and each QUBO transformation method, we have generated 100 samples. 

\begin{figure}[h]
\centering
\includegraphics[width=.49\textwidth]{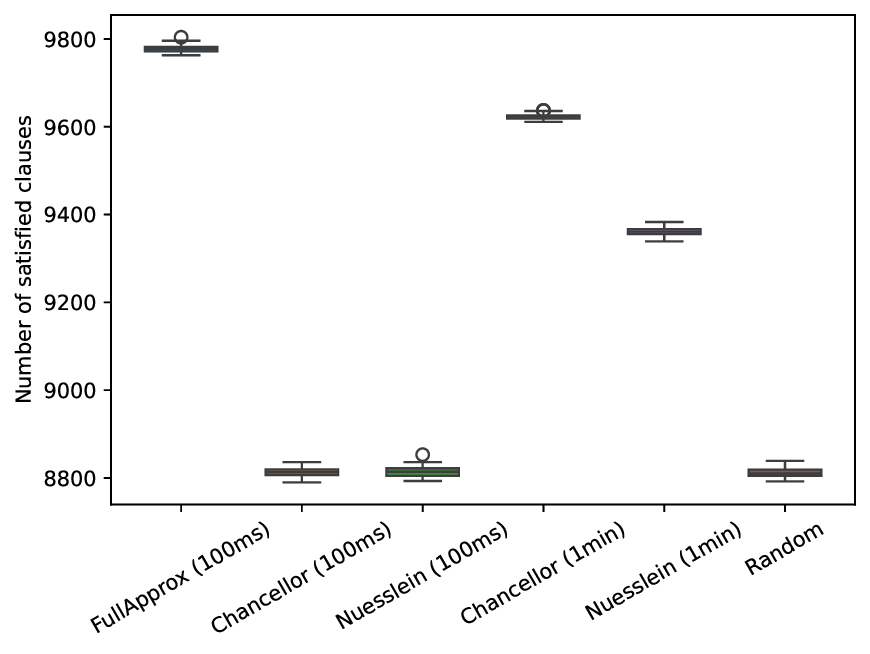}
\caption{Distribution of the best solutions found for each combination of a QUBO transformation and tabu solver time limit.}\label{fig:scaling_results}
\end{figure}

\noindent The results shown in Fig. \ref{fig:scaling_results} reveal two interesting insights. Firstly, we observe that randomly guessing solutions leads to approximately 8,800 satisfied clauses, which is close to the expected value of $7/8 \times 10,000$ clauses. Even though the size of the created formulas is 20 times larger than the size of the formulas created in Sec. \ref{sec:general_idea}, our full approximation method (FullApprox) still yields good solutions. For each formula, our FullApprox approach found solutions that satisfy approximately 98\% of all clauses. Secondly, to put these results into perspective, we compare the results of the FullApprox method to the results of Chancellor's and Nüßlein's method, also shown in Fig. \ref{fig:scaling_results}. We can see that for our dataset, the FullApprox method still yielded the best results.\\Observe, that for the FullApprox method, the tabu search method was given 100 ms of computation time. In contrast, the tabu search method was allowed 1 minute of computation time for the non-approximated exact QUBO transformations by Chancellor and Nüßlein. This is a consequence of D-Wave's implementation of the tabu search method. As the QUBOs of the non-approximated, exact QUBO transformations by Chancellor and Nüßlein are approximately 4.6 times as large as the QUBOs generated by our approximation method FullApprox, D-Wave's specific implementation of the tabu search method needs more time to generate reasonable results. We thus allowed the tabu search method 600 times more time to solve the 4.6 times larger QUBO matrices resulting from Chancellor's and Nüßlein's transformation. We want to emphasize that we included the specific time parameters for reproducibility in this paper. This is not meant as a performance comparison. By providing a different implementation of the tabu search method, one could reduce the solving times for larger QUBOs. The sole intention of this section, and especially the comparison of our approximated method with the non-approximated methods, is to show that our QUBO approximation for \msat problems still yields comparably good results, even if the problem size increases.

%% file: chapters/6-Conclusion.tex
\section{Conclusion}\label{sec:conclusion}
Due to the missing error correction of contemporary quantum hardware, any quantum computer calculation can be considered an involuntary solution approximation. In this work, we thus studied whether it is possible to improve the solution quality when solving problems on D-Wave's quantum annealer Advantage\_System6.4 by using systematic QUBO approximation of \msat problems instead of exact, non-approximated QUBO representations as an input for the quantum annealer. For a \msat instance consisting of a \tsat{} formula with $n$ variables and $m$ clauses, our proposed QUBO approximation method yields $(n\times n)$-dimensional QUBO matrices, which is considerably smaller than the QUBO matrices that result from any exact, non-approximated QUBO transformation. The method is based on an adaption of the creation method of exact QUBO transformations that result in QUBO matrices of dimension $(n+m) \times (n+m)$. In an empirical evaluation, we demonstrated that our QUBO approximations can yield comparable or even better results than exact, non-approximated QUBO transformations when solved on D-Wave's quantum annealer Advantage\_System6.4. Furthermore, 50\% larger \msat instances can be solved on the quantum annealer due to our approximation method's reduced need for physical qubits. Additionally, we empirically showed that naive methods of creating QUBO approximations for \msat problems using $(n+m)\times(n+m)$-dimensional QUBO matrices as initial QUBOs are not effective.\\ 
In the future, we would like to explore the cause for the different behavior of the naive QUO approximation methods \emph{min pruning} and  \emph{random pruning}. Furthermore, it is also interesting to investigate whether systematic QUBO approximation can be beneficial for solving other classes of hard problems on D-Wave's quantum annealer.